\documentclass[a4paper,USenglish,cleveref,autoref,nolineno]{socg-lipics-v2021}
\hideLIPIcs
\usepackage[utf8]{inputenc}
\usepackage{tikz}
\usepackage{tikz-cd}
\usetikzlibrary{matrix}
\usepackage{graphicx}
\usepackage{pgfplotstable}
\usepackage{amssymb}
\usepackage{mathtools}
\usepackage{amsmath} 
\usepackage{amsthm}
\usepackage{algorithm}
\usepackage{algpseudocode}
\usepackage{hyperref}
\usepackage{wrapfig}
\usepackage{listings}
\usetikzlibrary{positioning,matrix, arrows.meta}

\usepackage{xcolor,colortbl}
\pgfplotsset{compat=1.17}

\newcommand{\red}[1]{{\color{red}#1\normalcolor}}

\newcommand{\R}{\ensuremath{\mathbb{R}}}

\definecolor{Gray}{gray}{0.85}
\definecolor{LightCyan}{rgb}{0.88,1,1}

\newcolumntype{a}{>{\columncolor{Gray}}c}
\newcolumntype{b}{>{\columncolor{white}}c}

\newcommand{\marc}[1]{\marginpar{\small\color{olive}  #1}}

\title{Swap, Shift and Trim to Edge Collapse a Filtration}
\author{Marc Glisse}{Université Paris-Saclay, CNRS, Inria, Laboratoire de Mathématiques d'Orsay, 91405, Orsay, France.}{marc.glisse@inria.fr}{https://orcid.org/0000-0001-6914-1651}{}
\author{Siddharth Pritam}{Shiv Nadar University, School of Engineering, Department of Computer Science, Delhi NCR, India.}{siddharth.pritam@snu.edu.in}{https://orcid.org/0000-0001-5673-0406)}{}
\authorrunning{M. Glisse and S. Pritam}
\Copyright{Marc Glisse and Siddharth Pritam}
\date{November 2021}
\ccsdesc{Theory of computation~Computational geometry}
\ccsdesc{Mathematics of computing~Algebraic topology}
\keywords{edge collapse, flag complex, graph, persistent homology}
\supplementdetails[cite=Gudhi]{Software}{https://github.com/GUDHI/gudhi-devel}

\begin{document}
\algnewcommand{\algorithmicgoto}{\textbf{go to}}%
\algnewcommand{\Goto}[1]{\algorithmicgoto~\ref{#1}}%

\maketitle

\begin{abstract}
  Boissonnat and Pritam introduced an algorithm to reduce a filtration of flag (or clique) complexes, which can in particular speed up the computation of its persistent homology~\cite{FlagCompEdgeColl}. They used so-called \emph{edge collapse} to reduce the input flag filtration and their reduction method required only the $1$-skeleton of the filtration. In this paper we revisit the usage of edge collapse for efficient computation of persistent homology. We first give a simple and intuitive explanation of the principles underlying that algorithm. This in turn allows us to propose various extensions including a zigzag filtration simplification algorithm. We finally show some experiments to better understand how it behaves.
\end{abstract}
\section{Introduction}

Efficient computation of persistent homology has been a central quest in Topological Data Analysis (TDA) since the early days of the field about 20 years ago. Given a filtration (a nested sequence of simplicial complexes), computation of persistent homology involves reduction of a boundary matrix, whose rows and columns are the simplices of the input filtration. Traditionally, there are two complementary lines of research that have been explored to improve the computation of persistent homology. The first approach led to improvement of the persistence algorithm (the boundary matrix reduction algorithm) and of its analysis, to efficient implementations and optimizations, and to a new generation of software~\cite{Gudhi,phat,ripser,eirene,dionysus,ripser++,dory}. The second and complementary approach is to reduce (or simplify) the input filtration to a smaller filtration through various geometric or topological techniques in an exact or approximate way and then compute the persistent homology of the smaller reduced filtration. This research direction has been intensively explored as well~\cite{Mischaikow, PawelSimpleColl, ChazalOudot,Botnan, SheehyRipsComp, KerberSharath, Aruni, Simba}.

Flag complexes and, in particular, the Vietoris-Rips complexes are an important class of simplicial complexes that are extensively used in TDA. Flag complexes are fully characterized by their graph (or 1-skeleton) and can thus be stored in a very compact way. Therefore, they are of great practical importance and are well studied theoretically. Various efficient codes and reduction techniques have been developed for those complexes~\cite{ripser, SheehyRipsComp,ripser++,dory}. However, further progress have been made only recently by the work of Boissonnat and Pritam~\cite{FlagCompStrongColl,FlagCompEdgeColl}. Both works~\cite{FlagCompStrongColl, FlagCompEdgeColl} put forward preprocessing techniques, which reduce an input flag filtration (nested sequence of flag complexes) to a smaller flag filtration using only the $1$-skeleton. The work in~\cite{FlagCompStrongColl} uses a special type of collapse called strong collapse (removal of special vertices called dominated vertices), introduced by J. Barmak and E. Miniam~\cite{StrongHomotopy}.  In~\cite{FlagCompEdgeColl} they extend the notion of strong collapse to edge collapse (removal of special edges, called dominated edges) and use it for further filtration simplification which  improves the performance by several orders of magnitude.

In this paper, we revisit the usage of edge collapse for efficient computation of persistent homology. We first give a simple and intuitive explanation of the principles underlying the algorithm proposed in~\cite{FlagCompEdgeColl}. We identify that an algorithm to edge collapse a filtration can be deconstructed as three fundamental operations: 1. \textit{Swap} two edges having same filtration value, 2. \textit{Shift} a dominated edge forward in the filtration and 3. \textit{Trim} the very last dominated edge. This new approach allows us to propose various extensions, which we list below.

\begin{itemize}
    
    \item \textbf{Backward:} We propose a backward reduction algorithm, which processes the edges of a flag filtration with decreasing filtration values different to the algorithm in~\cite{FlagCompEdgeColl}. The algorithm in~\cite{FlagCompEdgeColl} processes edges one by one with increasing filtration values, i.e. in the forward direction. The backward processing results (shown experimentally) in faster reduction of the edges as it allows various operations like domination checks, computing the neighbourhood of an edge etc to be performed fewer times than in the forward algorithm of~\cite{FlagCompEdgeColl}. However, the forward algorithm of~\cite{FlagCompEdgeColl} has the advantage when the input filtration is in streaming fashion. Once we identify that to swap, to shift and to trim are the most basic operations of the reduction algorithm in~\cite{FlagCompEdgeColl}, it becomes clear that there could be possibly several different ways to reduce an input flag filtration using edge collapse. The forward algorithm of~\cite{FlagCompEdgeColl} and the backward algorithm proposed in this article are two natural variants of possibly several different variants one can think of.

    \item \textbf{Parallel: }We propose a divide and conquer heuristic to further improve and semi-parallelize our backward reduction algorithm. Our approach is to subdivide the input filtration into two smaller sub-sequences (consisting of consecutive edges), we process these smaller sub-sequences in parallel and then merge the solutions of two sequences to form the solution of the complete sequence. The two sub-sequences can be further sub-divided and processed recursively in parallel. 

    \item \textbf{Approximate: }With this simplified perspective a simple tweak in the backward algorithm allows us to have an approximate version of the reduction algorithm. There are two goals in mind behind an approximate version, first to speed up the algorithm, and second to obtain a smaller reduced sequence. We perform certain experiments to show how the approximate version performs on these two parameters. 
    
    \item\textbf{ Zigzag: }Next, we provide a reduction algorithm for a zigzag flag filtration, which is a sequence of flag complexes linked through inclusion maps however the inclusion maps could possibly be in both forward and backward directions.  The theoretical results in~\cite{FlagCompEdgeColl} can easily be extended to zigzag filtrations. We show that with the new point of view there is a simple algorithm for zigzag flag filtrations which incorporates parallelism as well.
\end{itemize}

We note that we don't assume that all the vertices appear in the beginning of the filtration. That is the filtration values of vertices can be arbitrary as well.

 On the theory side, we show that the edge collapse of a flag filtration can be computed in time $O(n_e\, k^3)$, where $n_e$ is the number of input edges and $k$ is the maximal degree of a vertex in the input graph. The algorithm  has been implemented and the code is available in the Gudhi library~\cite{Gudhi}.

An outline of this paper is as follows. Section~\ref{sec:prel} recalls some basic ideas and constructions related to simplicial complexes, persistent homology and collapses. We present the new simplified perspective and associated lemmas in Section~\ref{sec:shift_swap_trim}. In Section~\ref{sec:persistence_simplification}, we explain the new backward algorithm for flag filtration simplification. In Section~\ref{sec:parallel} and Section~\ref{sec:approximate}, we discuss the approach to parallel simplification and approximate computation respectively using edge collapse. The simplification algorithm for zigzag flag filtration is discussed in Section~\ref{sec:zigzag_filt}. Experiments are discussed in Section~\ref{sec:experiments}.

\section{Background}
\label{sec:prel}
In this Section, we briefly recall the basic notions like simplicial complexes, flag complexes, persistent homology and edge collapse. For more details on these topics please refer to~\cite{HerbertHarerbook, Hatcher, Munkres}.

\subparagraph{Simplicial complex and simplicial map.}

An \textbf{abstract simplicial complex} $\textit{K}$ is a collection of subsets of a non-empty finite set $\textit{X},$ such that for every subset $\textit{A}$ in $\textit{K}$, all the subsets of $\textit{A}$ are in $\textit{K}$. We call an \textit{abstract simplicial complex} simply a \textit{simplicial complex} or just a \textit{complex}. An element of $\textit{K}$ is called a \textbf{simplex}. An element of cardinality $k+1$ is called a $k$-simplex and 
$k$ is called its \textbf{dimension}.
Given a simplicial complex $K$, we denote its geometric realization as $|K|$.
A simplex is called \textbf{maximal} if it is not a proper subset of any other simplex in $\textit{K}$. A sub-collection  $\textit{L}$ of $\textit{K}$ is called a \textbf{subcomplex} if it is a simplicial complex itself.
An inclusion $\psi : K \xhookrightarrow{\sigma} K \cup \sigma$ of a single simplex $\sigma$ is called \textbf{elementary}, otherwise, it's called \textbf{non-elementary}. 
An inclusion $\psi : K \hookrightarrow L$ between two complexes $K$ and $L$ induces a continuous map $|\psi| : |K| \rightarrow |L|$ between the underlying geometric realizations. 

\subparagraph{Flag complex and neighborhood. } A complex $K$ is a \textbf{flag} or a \textbf{clique} complex if, when a subset of its vertices forms a clique (i.e.\ any pair of vertices is joined by an edge), they span a simplex. It follows that  the full structure of  $K$ is determined by its 1-skeleton (or graph) we denote by $G$.  For a vertex $v$ in $G$, the \textbf{open neighborhood} $N_G(v)$ of $v$ in $G$ is defined as  $N_G(v) := \{u \in G \: | \; [uv] \in E\}$, where $E$ is the set of edges of $G$. The \textbf{closed neighborhood} $N_G[v]$ is $N_G[v] := N_G(v) \cup \{ v\}$. 
Similarly we define the closed and open neighborhood of an edge $[xy] \in E$, $N_G[xy]$ and $N_G(xy)$ as $N_G[xy] := N_G[x] \cap N_G[y]$ and $N_G(xy) := N_G(x) \cap N_G(y)$, respectively.

\subparagraph{Persistent homology.}
A \textbf{sequence} of simplicial complexes $\mathcal{F}$ : $\{K_1 \hookrightarrow K_2 \hookrightarrow  \cdots \hookrightarrow  K_m \}$ connected through inclusion maps is called a \textbf{filtration}.
A filtration is a \textbf{flag filtration} if all the simplicial complexes $K_i$ are flag complexes.

If we compute the homology groups of all the $K_i$, we get the sequence $\mathcal{P}(\mathcal{F})$ : $\{H_p(K_1) \xhookrightarrow{*} H_p(K_2) \xhookrightarrow{*} \cdots \xhookrightarrow{*} H_p(K_m)\}$. Here $H_p()$ denotes the homology group of dimension $p$ with coefficients from a field $\mathbb{F}$ and $\xhookrightarrow{*}$ is the homomorphism induced by the inclusion map. $\mathcal{P}(\mathcal{F})$ is a sequence of vector spaces connected through the homomorphisms and it is called a \textbf{persistence module}. More formally, a \textit{persistence  module} $\mathbb{V}$ is a sequence of vector spaces $\{V_1 \xrightarrow{} V_2 \xrightarrow{} V_3 \xrightarrow{} \cdots \xrightarrow{} V_m\}$ connected with homomorphisms $\{\xrightarrow{}\}$ between them. A  persistence module arising from a sequence of simplicial complexes captures the evolution of the topology of the sequence. 

Any persistence module can be \textit{decomposed} into a collection of intervals of the form $[i,j)$~\cite{structure-pers}.
The multiset of all the intervals $[i, j)$ in this decomposition is called the \textbf{persistence diagram} of the persistence module. An interval of the form $[i,j)$ in the persistence diagram of $\mathcal{P}(\mathcal{F})$ corresponds to a homological feature (a `cycle') which appeared at $i$ and disappeared at $j$. The persistence diagram (PD) completely characterizes the persistence module, that is, there is a bijective correspondence between the PD and the equivalence class of the persistence module \cite{HerbertHarerbook, CarlssonZomorodian}.

Two different persistence modules $\mathbb{V} : \{V_1 \xrightarrow{} V_2 \xrightarrow{} \cdots \xrightarrow{} V_m\}$ and $\mathbb{W} : \{W_1 \xrightarrow{} W_2 \xrightarrow{} \cdots \xrightarrow{} W_m\}$, connected through a set of homomorphisms $\phi_i: V_i \rightarrow W_i$ are \textbf{equivalent} if the $\phi_i$ are isomorphisms and the following diagram commutes ~\cite{HerbertHarerbook, quivers}. Equivalent persistence modules  have the same interval decomposition, hence the same diagram.

  \begin{center}
   \begin{tikzcd}
  V_1 \arrow{r}{} \arrow{d}{\phi_1} & V_2 \arrow{r}{} \arrow{d}{\phi_2} & 
   \cdots &  \arrow{r} & V_{m-1} \arrow{r}{} \arrow{d}{\phi_{m-1}} & V_m \arrow{d}{\phi_m} \\
   W_1 \arrow{r}{} & W_2 \arrow{r}{} & \cdots & \arrow{r} & W_{m-1} \arrow{r}{} & W_m
   \end{tikzcd}
   \end{center}

\subparagraph{Edge collapse of a flag complex:}
In a flag complex $K$, we say that an edge $e =[ab]$, connecting vertices $a$ and $b$, is \textbf{dominated} by a vertex $v$ (different from $a$ and $b$) if $N_G[e]\subseteq N_G[v]$.
Removing $e$ and all its cofaces from $K$ defines a smaller flag complex $K^{\prime}$. It has been proven in~\cite{FlagCompEdgeColl} that when $e$ is dominated by a vertex of $K$, the inclusion $K^{\prime} \subset K$ induces an isomorphism between the homology groups of $K^{\prime}$ and $K$. This removal is called an \textbf{edge collapse}.

\section{Swapping, shifting and trimming} \label{sec:shift_swap_trim}

In this Section, we show three simple and fundamental operations that preserve the persistence diagram of a flag filtration:
1. Swapping any two edges with the same filtration value, 2. Shifting a dominated edge, and 3. Trimming a dominated edge at the end of the filtration. These operations can be combined to simplify a flag filtration. 

Before we proceed, we will fix some notations. Let $\{t_1, t_2, \cdots, t_n\}$ be a finite index set where $t_i \in \mathbb{R}$ and $t_i < t_j$ for $i < j$. For convenience, we may consider $t_{n+1} = \infty$. With each $t_i $ (called the \textit{filtration value} or \textit{grade}) we associate a graph $G_{t_i}$ 
such that $G_{t_i} \hookrightarrow G_{t_{i+1}}$ is an \textit{inclusion}, (not necessarily elementary) of edges. The flag complex of $G_{t_i}$ is denoted as $\overline{G}_{t_i}$ and we consider the associated flag filtration $\mathcal{F} : \overline{G}_{t_1} \hookrightarrow \overline{G}_{t_2} \hookrightarrow \cdots \hookrightarrow \overline{G}_{t_n}$. The edges in the set $E := \{e_{1}, e_{2}, \cdots e_{m} \}$ ($m \geq n$) are thus indexed with an order compatible with the filtration values.

\subparagraph{Swapping:} Inserting several edges at the same filtration value can be done in any order. We state this basic observation as the following lemma.

\begin{lemma} [Swapping Lemma] \label{lemma:swap}
Given a flag filtration $\{\overline{G}_{t_1} \hookrightarrow \overline{G}_{t_2} \cdots \hookrightarrow \overline{G}_{t_n}\}$, such that  ${G_{t_{i}}} \hookrightarrow {G_{t_{i+1}}}$ is a non-elementary inclusion. Then, the indices of the edges ${G_{t_{i+1}}} \setminus {G_{t_{i}}}$ could be assigned interchangeably. That is, swapping their order of insertion preserves the persistence diagram. 
\end{lemma}

\subparagraph{Shifting:} In a filtration, insertion of a dominated edge does not bring immediate topological change. Therefore, its insertion can be shifted until the next grade and possibly even further. 

\begin{lemma} [Shifting Lemma] \label{lemma:shift}
Let $e$ be a dominated edge in $G_{t_i}$ inserted at grade $t_i$. Then, the insertion of $e$ can be shifted by one grade to $t_{i+1}$ without changing the persistence diagram. In other words, the persistence diagrams of the original flag filtration $\mathcal{F} := \{\overline{G}_{t_1} \hookrightarrow \cdots \hookrightarrow \overline{G}_{t_i} \hookrightarrow \overline{G}_{t_{i+1}} \hookrightarrow \cdots \hookrightarrow \overline{G}_{t_n}\}$ and the shifted filtration $\{\overline{G}_{t_1} \hookrightarrow \cdots \hookrightarrow \red{\overline{{G}_{t_{i}} \setminus e}} \hookrightarrow \overline{G}_{t_{i+1}} \hookrightarrow \cdots \hookrightarrow \overline{G}_{t_n}\}$  are equivalent.
\end{lemma}

\begin{proof}

The proof follows from the commutativity of the following diagram, where all maps are induced by inclusions, and the fact that all vertical maps are isomorphisms. 
\begin{center}
   \begin{tikzcd}
   {H_p(\overline{G}_{t_{i-1}})} \arrow[r, hook] \arrow[d, equal] & {H_p(\overline{G}_{t_i})} \arrow[r, hook] \arrow[d, "{|r_i|*}", shift left=1.5ex] & {H_p(\overline{G}_{t_{i+1}})} \arrow[d, equal] \\
   {H_p(\overline{G}_{t_{i-1}})} \arrow[r, hook] & {H_p(\overline{{G}_{t_{i}} \setminus e})} \arrow[r, hook] \arrow[u, hook] %\arrow[l, red, "{|f^{sd}|}" blue]
   & {H_p(\overline{G}_{t_{i+1}})} 
   \end{tikzcd}
   \end{center}

This implies that the persistence diagrams of the sequences $\{\overline{G}_{t_1} \hookrightarrow \cdots \hookrightarrow \overline{G}_{t_i} \hookrightarrow \overline{G}_{t_{i+1}} \hookrightarrow \cdots \hookrightarrow \overline{G}_{t_n}\}$ and $\{\overline{G}_{t_1} \hookrightarrow \cdots \hookrightarrow \overline{{G}_{t_{i}} \setminus e} \hookrightarrow \overline{G}_{t_{i+1}} \hookrightarrow \cdots \hookrightarrow \overline{G}_{t_n}\}$ are equivalent, see~\cite[Theorem 4] {FlagCompEdgeColl} for more details. Here, $|r_i|*$ is the isomorphism between the homology groups induced by the retraction map (on the geometric realizations of the complexes) associated to the edge collapse.
\end{proof}

After an edge is shifted to grade $t_{i+1}$, it can leap frog the edges inserted at grade $t_{i+1}$ using the swapping lemma (Lemma~\ref{lemma:swap}) and can explore the possibility of being shifted to the next grade.

\subparagraph{Trimming:} If the very last edge in the filtration is dominated then we can omit its inclusion. This is a special case of the shifting operation (Lemma~\ref{lemma:shift}) assuming that there is a graph $G_{\infty}$ at infinity.
\begin{lemma} [Trimming Lemma]\label{lemma:trimming}
Let $e\notin G_{t_{n-1}}$ be a dominated edge in the graph $G_{t_n}$. Then, the persistence diagrams of the original sequence $\mathcal{F} := \{\overline{G}_{t_1} \hookrightarrow \overline{G}_{t_2} \hookrightarrow \cdots \hookrightarrow \overline{G}_{t_n}\}$ and the trimmed sequence $\{\overline{G}_{t_1} \hookrightarrow \overline{G}_{t_2} \hookrightarrow  \cdots \hookrightarrow \overline{{G}_{t_{n}} \setminus e}\}$ are equivalent.
\end{lemma}
Note that when shifting or trimming produces a sequence with identical consecutive graphs $G_{t_i}=G_{t_{i+1}}$, we can just drop index $t_{i+1}$.
\begin{comment}
The following lemma from~\cite{FlagCompEdgeColl} justifies the fact that during the forward domination check (the internal forward loop (Line 12) of \cref{alg:core_flag_filtration}) we only include the edges from the edge-neighbourhood of an edge whose domination is being checked.
\end{comment}
\begin{lemma}[Adjacency] \label{lemma:nbd_domination}
 Let $e$ be an edge in a graph $G$ and let $e^\prime$ be a new edge with $G^\prime := G \cup e^\prime$.
 If $N_{G}(e)=N_{G'}(e)$ and $e$ is dominated in $G$, then $e$ is also dominated in $G'$.
\end{lemma}
This is in particular the case if $e$ and $e'$
are not boundary edges of a common triangle in $\overline{G^{\prime}}$. The above lemma is not strictly necessary, but it is useful to speed up algorithms.

 In the next Section, we show that one can cook up an algorithm to edge-collapse a flag filtration using these simple ingredients.

\section{Persistence simplification}\label{sec:persistence_simplification}
In this Section, we will describe our new approach to use edge collapse to speed up the persistence computation. As mentioned before, the simplification process will be seen as a combination of the basic operations described in Section~\ref{sec:shift_swap_trim}. This new perspective simplifies the design process and correctness proof of the algorithm. Along with this we achieve a significant improvement in the run-time efficiency as shown in Section~\ref{sec:experiments}. We first briefly look at the forward algorithm of~\cite{FlagCompEdgeColl} with this new point of view and then present the new approach called the \textit{backward algorithm} [Algorithm~\ref{alg:core_flag_filtration}]. Both algorithms take as input a flag filtration $\mathcal{F}$ represented as a sorted array $E$ of edges (pairs of vertices) with their filtration value, and output a similar array $E^c$, sorted in the case the Forward Algorithm and unsorted for the Backward Algorithm, that represents a reduced filtration $\mathcal{F}^c$ that has the same persistence diagram as $\mathcal{F}$.

\subparagraph{Forward algorithm.}
In the forward algorithm (the original one from~\cite{FlagCompEdgeColl}), the edges are processed in the order of the filtration in a streaming fashion. If a new edge is dominated, we skip its insertion and consider the next edge. If the next edge is dominated as well its insertion is skipped as well. Intuitively, the sequence of such dominated edges forms a train of dominated edges that we are moving to the right. When a new edge $e$ is  non-dominated (called \emph{critical}), we output it, and also check what part of the train of dominated edges is allowed to continue to the right (shifted forward) and what part has to stop right there. For all the previously dominated edges (actually only those that are \textit{adjacent} to $e$), we check if they are still dominated after adding the edge $e$. If an edge $e^{\prime}$ becomes critical, we output it with the same filtration value as $e$, and the following edges now have to cross both $e$ and $e^{\prime}$ to remain in the train. We stop after processing the last edge, and the edges that are still part of the dominated train are discarded (trimmed).

\subparagraph{Backward algorithm.}
The backward algorithm (\cref{alg:core_flag_filtration}) considers edges in order of decreasing filtration value. Each edge $e$ is considered once, delayed (shifted) as much as possible, then never touched again. We always implicitly swap edges so that while $e$ is the edge considered, it is the last one inserted at its current filtration value, and compute its domination status there. If the edge is dominated, we shift it to the next filtration value, and iterate, swapping and checking for domination again at this new filtration value. If there is no next filtration value, we remove the edge (trimming). Once the edge is not dominated, we update its filtration value and output it. As an optimization, instead of moving the edge one grade at a time, we may jump directly to the filtration value of the next adjacent edge, since we know that moving across the other edges will preserve the domination (\autoref{lemma:nbd_domination}).

The main datastructure used here is a neighborhood map $N$. For each vertex $u$, it provides a map $N[u]$ from the adjacent vertices $v_i$ to the filtration value $N[u][v_i]$ of edge $uv_i$. The two main uses of this map are computing the neighborhood of an edge $uv$ at a time $t$ (i.e. in the graph $G_t$) as $N_t[uv]=N_t[u]\cap N_t[v]$ (filtering out the edges of filtration value larger than $t$), and checking if such an edge neighborhood is included in the neighborhood of a vertex $w$ at time $t$. While computing $N_t[uv]$, we also get as a side product the list of the future neighbors $F_t[uv]$, i.e.\  the elements of $N_\infty[uv]\setminus N_t[uv]$, which we sort by filtration value. These operations can be done efficiently by keeping the maps sorted, or using hashtables. The information in $N$ is symmetric, any operation we mention on $N[u][v]$ (removal or updating $t$) will also implicitly be done on $N[v][u]$.
\begin{comment}

We first present the underlying principles of the algorithm. As mentioned before, the central idea is to delay the insertion of a dominated edge until it becomes non-dominated. This delay is achieved by increasing the index (filtration value) of the edge to a value when it becomes non-dominated. More specifically, we consider the edges in the decreasing filtration order as we move backward with the filtration value. When considering an edge $e_{i}$, we check whether it is dominated in $G_{t_i}$. If $e_{i}$ is dominated in $G_{t_i}$ then clearly the insertion of $e_{i}$ does not change the topology of $G_{t_{i-1}}$ and $e_{i}$ does not change the persistence diagram. Therefore, we can temporarily put its insertion on hold. Next, we check if the edge $e_{i}$ can be dominated in the graph $G_{t_{i+1}}$. If it is still dominated we still keep its insertion on hold and move to the next graph $G_{t_{i+2}}$. We repeat this process until $e_{i}$ is found to be non-dominated in some graph $G_{t_j}$ for some $j>i$. Suppose $e_{i}$ is found to be non-dominated in the graph $G_{t_j}$, then we set the new index of $e_{i}$ to be $t_j$ and denote it by $e_{i}^{t_j}$. On the other hand, if $e_{i}$ was found to be non-dominated in the graph $G_{t_i}$ at the first place then it is inserted right there by assigning its new index to $t_i$ (i.e. keeping the original index) and is denoted by $e_{i}^{t_i}$.
\end{comment}
In this Section, we denote $t(e)$ the filtration value of $e\in E$, which is stored as $N[u][v]$ if $e=uv$. Note that even though $E$ is sorted, since several edges may have the same filtration value, $G_{t(e)}$ may contain some edges that appear after $e$.

We now explain the precise computation of the reduced sequence of edges $E^c$. See [Algorithm~\ref{alg:core_flag_filtration}] for the pseudo-code. The main \lstinline{for} loop on line 4 (called the backward loop) iterates over the edges in the sequence $E$ by decreasing filtration values, i.e. in the \textit{backward direction},  and checks whether or not the current edge $e$ is dominated in the graph $G_{t(e)}$. If \textit{not},  we insert $e$ in $E^c$ and
keep its original filtration value ${t(e)}$. 
Else, $e$ is dominated in $G_{t(e)}$, and we increase $t(e)$ to the smallest value $t'>t(e)$ where $N_{t(e)}[e]\subsetneq N_{t'}[e]$. We can then iterate (\lstinline{goto} on line 12), check if the edge is still dominated at its new filtration value $t'$, etc. When the edge stops being dominated, we insert it in $E^c$ with its new $t(e)$ and update $t(e)$ in the neighborhood map $N$.
If the smallest value $t'>t(e)$ does not actually exist, we remove the edge from the neighborhood map and do \emph{not} insert it in $E^c$.

\begin{algorithm}[h]
    \caption{Core flag filtration backward algorithm}
    \label{alg:core_flag_filtration}
    \begin{algorithmic}[1] % The number tells where the line numbering should start
        \Procedure{Core-Flag-Filtration}{$E$}
        \State {\bf input :} set of edges $E$ sorted by filtration value
		\State $E^{c} \gets \emptyset$
	     \For{ $e \in E$}  \Comment {In non-increasing order of $t(e)$}
	        \State Compute $N_{{t(e)}}(e)$ and $F_{t(e)}(e)$
	        \For{$w\in N_{{t(e)}}(e)$}\label{line:loop1}
	            \State Test if $w$ dominates $e$ at $t(e)$
			\EndFor
			\If{ $e$ is dominated in $G_{t(e)}$}
			    \If{$F_{t(e)}(e)$ is empty}
			        \State Remove $N[u][v]$ \Comment{Trimming.}
    				\State \Goto{line:end-edge-loop} (next edge)
			    \Else \Comment{Shift and Swap.}
			        \State $t' \gets$ filtration of the first element of $F_{t(e)}(e)$ 
			        \State Move from $F_{t(e)}(e)$ to $N_{{t(e)}}(e)$ the vertices that become neighbors of $e$ at $t'$
			        \State $N[u][v]=t(e) \gets t'$
			        \State \Goto{line:loop1}
			    \EndIf
			\Else
				\State Insert $\{e,{t(e)}\}$ in $E^{c}$
				\State \Goto{line:end-edge-loop} (next edge)
			\EndIf	
		\EndFor\label{line:end-edge-loop}
		
		\State \textbf{return} $E^{c}$ \Comment{$E^{c}$ is the 1-skeleton of the core flag filtration.}
        \EndProcedure
    \end{algorithmic}
\end{algorithm}

\begin{theorem}[Correctness] \label{equivalence_thm}
Let $\mathcal{F}$ be a flag filtration, and $\mathcal{F}^c$ the reduced filtration produced by \cref{alg:core_flag_filtration}.
$\mathcal{F}$ and $\mathcal{F}^c$ have the same persistence diagram.
\end{theorem}

\begin{proof}
The proof is based on the observation that the algorithm inductively performs the elementary operations from \cref{sec:shift_swap_trim}: either it trims the very last edge of the current sequence (Line~11) or shifts and swaps a dominated edge forward to get a new sequence. Then the result follows using \cref{lemma:shift,lemma:swap,lemma:trimming} inductively. The only subtlety is around Line~15, where instead of simply performing one shift to the next filtration value, we perform a whole sequence of operations. We first shift $e$ to the next filtration value $t'$ (and implicitly swap $e$ with the other edges of filtration value $t'$). As long as we have not reached the first element of $F_{t(e)}(e)$, we know that shifting has not changed the neighborhood of $e$ and thus by \cref{lemma:nbd_domination} the fact that $e$ is dominated. We can then safely keep shifting (and swapping) until we reach that first element of $F_{t(e)}(e)$.
\end{proof}

\subparagraph{Complexity.} We write $n_e$ for the total number of edges and $k$ for the maximum degree of a vertex in $G_{t_n}$.
The main loop of the procedure, Line~4 of \cref{alg:core_flag_filtration}, is executed $n_e$ times. Nested, we loop (in the form of \lstinline{go to 6}) on the elements of $F_{t(e)}(e)$, of which there are at most $k$. For each of those elements, on Line~6, we iterate on $N_{t(e)}(e)$, which has size at most $k$. Finally, testing if a specific vertex dominates a specific edge amounts to checking if one set is included in another, which takes linear time in $k$ for sorted sets or hash tables. The other operations are comparatively of negligible cost. Sorting $F_{t(e)}(e)$ on Line~5 takes time $k\log k=o(k^2)$. Line~15 may take time $k\log k$ depending on the datastructure, $O(k^2)$ in any case. This yields a complexity of $O(n_ek^3)$.

Note that this analysis is pessimistic. If we delay an edge a lot (costly) but end up removing it, it makes future operations cheaper. Non-dominated edges only cost $k^2$. The edges that have many neighbors (usually appear late towards the end and) have few extra adjacent edges left to cross. After shifting (\lstinline{go to 6}), we can cheaply check if the previous dominator is still valid and in many cases skip the costly ($k^2$) full domination check.

\subparagraph{Optimality.}
The sequence produced by the backward (or forward) algorithm may still contain dominated edges, and \cref{sec:exp-complete} shows that it can take several runs of the algorithm before we end up with a fully reduced sequence. While each edge in the output was non-dominated at some point in the algorithm, other edges may later swap with this one and make it dominated again. It would be possible to enhance the algorithm so it performs some specific action when swapping an edge with a neighboring edge, the simplest being to mark it so we know this edge is worth checking in the next run, but one run of the backward algorithm already brings most of the benefit, so we did not concentrate our effort on a full reduction.

Swapping, shifting and trimming may produce many different reduced sequences. There is no reason to believe that our gready approach yields the smallest possible sequence, finding that sequence looks like a hard problem. However, we are mostly interested in finding a small enough sequence, so this is not an issue.

\section{Parallelisation}\label{sec:parallel}
Delaying the insertion of an edge until the next grade, and possibly swapping it, is a very local operation. As such, there is no problem doing several of them in parallel as long as they are in disjoint intervals of filtration values. 
We exploit this observation and further optimize our algorithm by parallelizing a significant part of the computation using a divide and conquer approach. 

To describe the parallel approach, let us use the same notations $t_i$, $G_{t_i}$, $\mathcal{F}$, $G_{\mathcal{F}}$ and $E$ as in Section~\ref{sec:shift_swap_trim}. To make things simpler, we assume that all edges have distinct filtration values. We subdivide the given input edge set $E := \{e_{1}, e_{2}, \cdots e_{n} \}$ of size $n$ into two smaller halves: the left half $E_l := \{e_{1}, e_{2}, \cdots e_{{n/2}} \}$ and the right half $E_r := \{e_{{n/2 + 1}}, e_{{n/2 +2}}, \cdots e_{{n}}\}$ of roughly the same size. We will describe a version of the algorithm based on the backward algorithm, but the same could be done with the forward algorithm, or they could even be mixed.

We first apply the backward algorithm to $E_l$ normally (\emph{left call}), which produces a reduced $E^c_l$. We also remember the list of all edges that were removed in this procedure: $E^r_l:=E_l\setminus E^c_l$. Independently (in parallel), we apply the backward algorithm to $E$ (\emph{right call}), but stop after processing all the edges of $E_r$ on Line~4 of \cref{alg:core_flag_filtration}. In a final sequential merging step, we resume the right call, processing only the edges of $E^r_l$, as if they all had the same initial filtration value $t_{n/2+1}$. The subdivision can obviously be applied recursively to increase the parallelism.

\begin{lemma}
The parallel algorithm produces exactly the same output as the sequential algorithm, and is thus correct.
\end{lemma}

\begin{proof}
The right call and the sequential algorithm start by handling the edges of $E_r$ in exactly the same way. When we reach the edges of $E_l$, for each edge $e$, there are two cases. Either the sequential algorithm shifts $e$ no further than $t_{n/2}$, in which case the left call does the same. Or the sequential algorithms shifts $e$ further (possibly all the way to removing it), then shifting to $t_{n/2+1}$ is handled by the left call, while the rest of the shift happens in the merging step.
\end{proof}

\section{Approximation} \label{sec:approximate}
Another interesting extension is an approximate version that gives a diagram within bottleneck distance $\epsilon$ of the true diagram (or some other similar criterion). Since the Rips filtration is often used as an approximation of the \v Cech filtration, an additional error is often acceptable.

If an edge is non-dominated for a short range of filtration values and becomes dominated again afterwards, it is tempting to skip the non-dominated region and keep delaying this edge. However, if we are not careful, the errors caused by these skips may add up and result in a diagram that is far from the original. The simplest idea would be to round all filtration values to the nearest multiple of $\epsilon$ before running the backward algorithm (similarly to~\cite{FlagCompStrongColl}). However, we can do a little better.

We describe here one safe approximation algorithm, based on the backward algorithm.
When considering a new edge $e$, instead of checking if it is dominated at its original position $t(e)$, we start checking $\epsilon$ later, at filtration value $t(e)+\epsilon$. If it is dominated, we resume normal processing from there.
However, if the edge is not dominated $\epsilon$ after its original insertion time, we keep it at its original position, so we don't end up uselessly shifting the whole sequence.
\begin{lemma}
  The resulting module is $\epsilon$-interleaved\footnote{See \cite{structure-pers} for a definition of interleaving.} with the original one.
\end{lemma}
\begin{proof}
Consider the set $D$ of edges that are delayed by this algorithm, and $C$ the edges that are kept at their original position.
Starting from the original sequence, we can delay all the edges of $D$ by exactly $\epsilon$. The flag filtration defined by this delayed sequence is obviously $(0,\epsilon)$-interleaved with the original.
We now run the regular backward algorithm on this sequence, with the difference that the edges in $C$ are handled as if they were never dominated. The output filtration has the same persistence diagram as the delayed sequence, which is at distance at most $\epsilon$ from the diagram of the original filtration.

The key observation here is that this procedure, where we first delay some edges then run the exact algorithm, produces the same output as the approximation algorithm.
\end{proof}
%\sid{The last sentence of the proof, ``The key observation...'' is not clear to me. What are the two different approximation algorithm, second and first?}
%\marc{Is it slightly better now?}
Many versions of this can be considered, with the goal to enable more reductions, but one needs to ensure that the $\epsilon$-approximations for two edges cannot combine to make an error larger than $\epsilon$ on the output. In this example, processing the edges from right to left is crucial to guarantee a bounded error, the approximation done with the initial shift from $t(e)$ to $t(e)+\epsilon$ of an edge $e$ cannot affect an already modified persistence interval, since those are after $t(e)+\epsilon$. However, it also limits the optimization opportunities a lot. It could make sense to run the exact simplification algorithm first, and only then make a pass with the approximation algorithm, to avoid "wasting" our approximation budget on unnecessary delays, but that would lose the advantage that the approximation algorithm may be faster than the exact one.
%\rvw{the message here is not very clear to me. In which cases could one lose control of the approximation? Why doesn't it happen here? The authors may want to consider a reformulation of the paragraph}
\section{Zigzag persistence} \label{sec:zigzag_filt}
The filtrations we have discussed so far are increasing sequences of complexes. There exists a more general type of filtration, called zigzag filtration~\cite{CarlssonZigzag,ClementSteveZigzag} $ \mathcal{Z}: K_1 \hookrightarrow  K_2 \hookleftarrow K_3 \hookrightarrow \cdots  \hookrightarrow K_n $. Here consecutive complexes are still related by an inclusion, but the direction of this inclusion may be different for every consecutive pair. In other words, as we move from left to right with increasing indices, the complex $K_i$ is obtained by either \textit{inclusion} of simplices or \textit{removal} of simplices from the previous complex $K_{i-1}$. Persistence diagrams can still be defined for these filtrations. Again, in this paper, we are only interested in flag zigzag filtrations, where each complex is a clique complex. For a flag zigzag filtration the underlying graphs are related through inclusions or removals of edges. We show that edge collapse can again be used for simplification of such sequences.

In the case of standard persistence (explained in Section~\ref{sec:persistence_simplification}) the goal of the simplification process was to shift as many dominated edges as possible towards the end of a filtration and then trim them. For a zigzag flag filtration there are several possible ways to simplify it: 1. If a dominated edge is included and is never removed, then as usual we try to \textit{shift} it towards the end and trim it. 2. If an edge is included and removed both as dominated, then we try to \textit{shift} the inclusion till its removal and then annihilate both operations. 3. If an edge is included as non-dominated but later removed as dominated then we try to shift its removal towards the right till the end or its re-insertion. 
4. A zigzag filtration is symmetric and a removal is an inclusion from the opposite direction, therefore, we can shift dominated removals towards the beginning and perform symmetric operations as in 2.

The 3rd method reduces the number of events at the cost of a slightly bigger complex, which may or may not be preferred over a more ``zigzagy'' filtration, so we do not use it in the default algorithm. 

With more ways to simplify, the simplification process of a zigzag flag filtration is more delicate compared to the usual filtration. And it has some subtleties, first, can we shift a dominated edge inclusion across an edge removal? We show that (in \cref{lemma:zigzag_shift}), a dominated edge $e$ can be shifted across an edge removal if $e$ is also dominated after the edge removal. Resolving the first issue leads us to the question, how to index (order) inclusions and removals of the same grade? In practice, this situation is not common and two complexes at consecutive grades are linked through either inclusions or removals.
Therefore, we adopt the following representation for a zigzag flag filtration.

We will use the same notations $t_i$, $G_{t_i}$, $\overline{G}_{t_i}$ and $E$ as in Section~\ref{sec:shift_swap_trim}.
We represent a zigzag filtration in slightly more general way as $\mathcal{Z} : \overline{G}_{t_1} \hookleftarrow \overline{G}_{t'_1} \hookrightarrow \overline{G}_{t_2} \hookleftarrow \cdots \hookrightarrow \overline{G}_{t_{i-1}} \hookleftarrow \overline{G}_{t_{i-1}^{\prime}} \hookrightarrow
\overline{G}_{t_i} \hookleftarrow \overline{G}_{t_{i}^{\prime}} \hookrightarrow
\overline{G}_{t_{i+1}},\cdots \hookrightarrow \overline{G}_{t_n}$. Here ${G}_{t_{i}^{\prime}}$ is an intermediate graph at grade $t_i$.
In a usual zigzag, $\overline{G}_{t'_{i}}$ is equal to either $\overline{G}_{t_i}$ or $\overline{G}_{t_{i+1}}$ depending on the direction of the arrow. Note that the standard zigzag algorithm still applies to this version.

The following lemma provides a sufficient condition to shift and swap an inclusion with removal.

\begin{lemma} [The Zigzag Shifting-Swapping Lemma] \label{lemma:zigzag_shift}
Let $e$ be an edge inserted at $t_i$, $e \in G_{t'_{i}}$ and dominated in both graphs $G_{t_i}$ and $G_{t'_{i}}$.
Then the persistence diagrams of the original zigzag flag filtration 
$\{\ \cdots \hookleftarrow\overline{G}_{t'_{i-1}}\xhookrightarrow{e} \overline{G}_{t_i} \hookleftarrow \overline{G}_{t_{i}^{\prime}} \hookrightarrow \overline{G}_{t_{i+1}} \hookleftarrow \cdots\ \}$
and the shifted-swapped sequence 
$\{\ \cdots \hookleftarrow\overline{G}_{t'_{i-1}}\hookrightarrow \overline{G_{t_i}\setminus e} \hookleftarrow \overline{G_{t_{i}^{\prime}}\setminus e} \xhookrightarrow{e} \overline{G}_{t_{i+1}} \hookleftarrow \cdots \ \}$
are equivalent. That is, the grade of $e$ can be shifted to $t_{i+1}$. 
\end{lemma}
\begin{proof}
The proof follows through a similar argument as  \cref{lemma:shift}. All three squares in the following diagram commute as all the maps are induced by inclusions. Note that the top left and the bottom right horizontal maps can be induced by the inclusion of more edges than just $e$.

\begin{center}
   \begin{tikzcd}
    H_p(\overline{G}_{t_{i-1}^{\prime}})   \arrow[r, hook, "{|e|^*}"] & H_p(\overline{G}_{t_i}) \arrow[d, "{|rt|*}", shift left=1.5ex] & H_p(\overline{G}_{t_{i}^{\prime}}) \arrow[l, hook, blue ] \arrow[r, hook] \arrow[d, "{|rt|*}", shift left=1.5ex] & H_p(\overline{G}_{t_{i+1}}) \arrow[d, equal] \\
    H_p(\overline{G}_{t_{i-1}^{\prime})} \arrow[u, equal] \arrow[r, hook] &
   H_p(\overline{G_{t_{i}} \setminus e})  \arrow[u, hook, "{|e|^*}"]  
   & H_p(\overline{G_{t_{i}^{\prime}} \setminus e})  \arrow[l, hook, red]  \arrow[u, hook, "{|e|^*}"] \arrow[r, hook ,"{|e|^*}"] & H_p(\overline{G}_{t_{i+1}}) 
   \end{tikzcd}
\end{center}
Since the vertical maps are either equalities or isomorphisms induced by the inclusion of the dominated $e$ ($|rt|$ is the corresponding retraction map associated with the collapse), the result follows immediately. That is, the shift of $e$ to the grade $t_{i+1}$ preserves the diagram. 
\end{proof}

Note that in the above lemma, the hypothesis that the edge $e$ should be dominated in  the graph $G_{t_{i}^\prime}$ is necessary as shown in \cref{fig:zigzag_example}. 

\begin{figure}[H]   
\centering
\includegraphics[scale=.6]{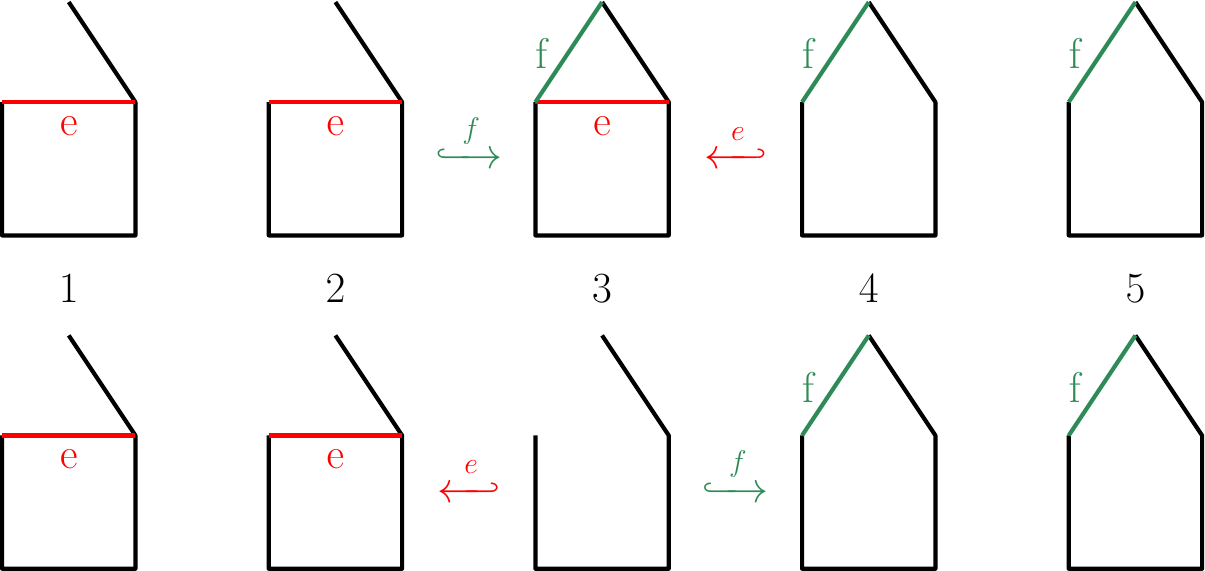}
\caption{In the top sequence, the green edge $f$ is dominated at grade $3$ and non-dominated at grade $4$. Shifting and swapping the inclusion of $f$ with the removal of the red edge $e$ results in the bottom sequence. This results in two different one dimensional persistence diagrams of the associated flag complexes.  For the top sequence it is  $\{[1,5]\}$ and for the bottom $\{[1,2], [4,5]\}$. Note that it is standard to use closed intervals in a zigzag persistence diagram.}
\label{fig:zigzag_example} 
\end{figure}

If a dominated edge is inserted and removed at the same grade, we can cancel both operations. 

\begin{lemma}[The Cancellation Lemma] \label{lemma:zigzag_cancel}
Let $e$ be an edge inserted and removed at $t_i$. If $e$ is dominated in $G_{t_i}$, then the persistence diagrams of the following two sequences 
$\{\ \cdots \hookleftarrow\overline{G}_{t'_{i-1}}\xhookrightarrow{e} \overline{G}_{t_i} \xhookleftarrow{e} \overline{G}_{t_{i}^{\prime}} \hookrightarrow \overline{G}_{t_{i+1}} \hookleftarrow \cdots\ \}$ and
$\{\ \cdots \hookleftarrow\overline{G}_{t'_{i-1}} \hookrightarrow \overline{G_{t_i}\setminus e} \hookleftarrow \overline{G}_{t_{i}^{\prime}} \hookrightarrow \overline{G}_{t_{i+1}} \hookleftarrow \cdots\ \}$ are the same.
\end{lemma}

\subparagraph{Algorithm: }

The algorithm to simplify $\mathcal{Z} : \overline{G}_{t_1} \hookleftarrow  \cdots 
\overline{G}_{t_i} \hookleftarrow \overline{G}_{t_{i}^{\prime}} \hookrightarrow
\overline{G}_{t_{i+1}},\cdots \hookrightarrow \overline{G}_{t_n}$ is again a combination of swapping, shifting and trimming of a dominated edge. For each edge $e$ in $\mathcal{Z}$ there is a list of pairs $<t, inc>$ associated with it, where $t$ is a grade and \textit{inc} is a Boolean variable to denote whether $e$ is inserted or removed at $t$.
Below, we provide the main steps of the zigzag simplification algorithm. The algorithm first processes all the edge inclusions in decreasing grade order from $t_n$ to $t_1$ and tries to shift them towards the end. After processing the first edge inclusion, it processes all the removals in increasing grade order from $t_1$ to $t_n$ and tries to shift them towards the beginning. This process can be repeated several times until it converges. We use $t(e)$ to denote the current grade of the edge $e$ being considered by the algorithm.

\begin{algorithm}[H]
	\caption{Core zigzag flag filtration algorithm}
	\label{alg:core_zigzag_algorithm}
	\begin{algorithmic}[1] % The number tells where the line numbering should start
		\Procedure{Core-Zigzag-Flag-Filtration}{$E$}
		\ForAll {edge \textit{inclusions}, backward (from $t_n$ to $t_1$)}
		\If {the current edge $e$ is dominated in the graph $G_{t(e)}$} 
		\If {$t{(e)} == t_n$}
		\State trim $e$ (delete the element $<t(e), inc>$).
		\ElsIf{ $G_{t(e)} \neq G_{{t'(e)}}$}  \Comment{the next step is a removal $G_{t(e)}  \hookleftarrow G_{{t'(e)}}$}.
		\If{$e \notin G_{t'(e)}$}
		\State delete the inclusion-removal pair of $e$ at $t(e)$.
		\ElsIf{ $e$ is dominated in $G_{t'(e)}$.} 
		\State set $t(e) = t{(e)}+1$ and go-to step 3. \Comment{$t{(e)}+1$ denotes the next grade.}
		\EndIf
		\Else \Comment{the next step is an inclusion $G_{t(e)}  \hookrightarrow G_{t(e)+1}$.}
		\State set $t(e) = t(e)+1$  and go-to step 3.
		\EndIf
		\EndIf
		\EndFor
		\State Move forward from $t_1$ to $t_n$ and process edge \textit{removals} symmetric to steps 2-16. \label{alg-zig-rev}
		
		\EndProcedure
	\end{algorithmic}
\end{algorithm} 

We skip the details of Line~\ref{alg-zig-rev}, which is similar to the previous loop.
Note that an edge can be inserted and removed multiple times, in this case, the algorithm proceeds by pairing an inclusion with its next removal.

The above algorithm outlines the essential aspects of the computation but is not optimal. Like \cref{alg:core_flag_filtration} we can use the Adjacency lemma (\cref{lemma:nbd_domination}) to perform fewer domination checks.
In an oscillating rips complex~\cite{RipsZigzagOudotSheehy}, it is quite common for edges to be transient (appear and disappear almost immediately). Identifying such edges and getting rid of both their insertion and removal is the hope of this simplification process. 
\subparagraph{Correctness: }As the underlying principal of the zigzag algorithm is the same as Algorithm~\ref{alg:core_flag_filtration}. We avoid its detailed discussion. To certify the correctness of the algorithm, again we observe the fact that the above simplification algorithm inductively performs either shifting, swapping or trimming of a dominated edge that are validated by \cref{lemma:swap,lemma:zigzag_shift,lemma:trimming}. Note that \cref{lemma:swap,lemma:trimming} extend naturally to the zigzag case.

We can easily parallelize the zigzag simplification algorithm using the same divide and conquer approach described in Section~\ref{sec:parallel}.

\section{Experiments}
\label{sec:experiments}
In this Section we provide various set of experiments to showcase the efficiency of our new approach. We also benchmark the new approach with the current state of the art methods.
\subparagraph{Complete graph.}
\label{sec:exp-complete}
Starting from a complete graph on 700 vertices where all edges appear at the same time, the size of the graph after applying the algorithm several times decreases as 244650 (initial), 5340, 3086, 1307, 788 and finally 699. It stops decreasing after the 5th round since 699 edges is obviously minimal. This example demonstrates that one round of the algorithm is far from producing a fully reduced sequence. However, it removed a large number of edges, which makes subsequent rounds much faster, and may have already reduced the complex enough to compute (persistent) homology.

\subparagraph{Torus: distribution of filtration values.}
We use a dataset with 1307 points on a torus embedded in $\R^3$. Figure~\ref{fig:torus} (left) shows the distribution of the edge lengths. Originally, there are 853471 edges and the longest has size $2.6$. We apply repeatedly the backward algorithm until the process converges. In the end, we are left with 65053 edges, and a maximal filtration value of $1.427$.
\begin{figure}[H]
    \centering
    \includegraphics[width=.5\textwidth]{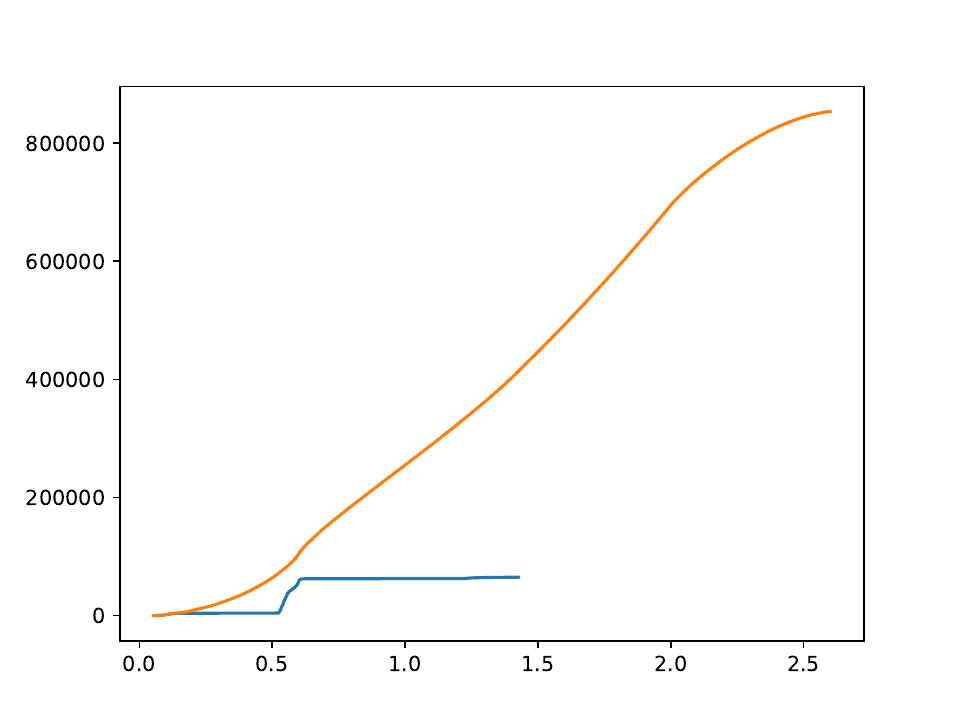}%
    \includegraphics[width=.5\textwidth]{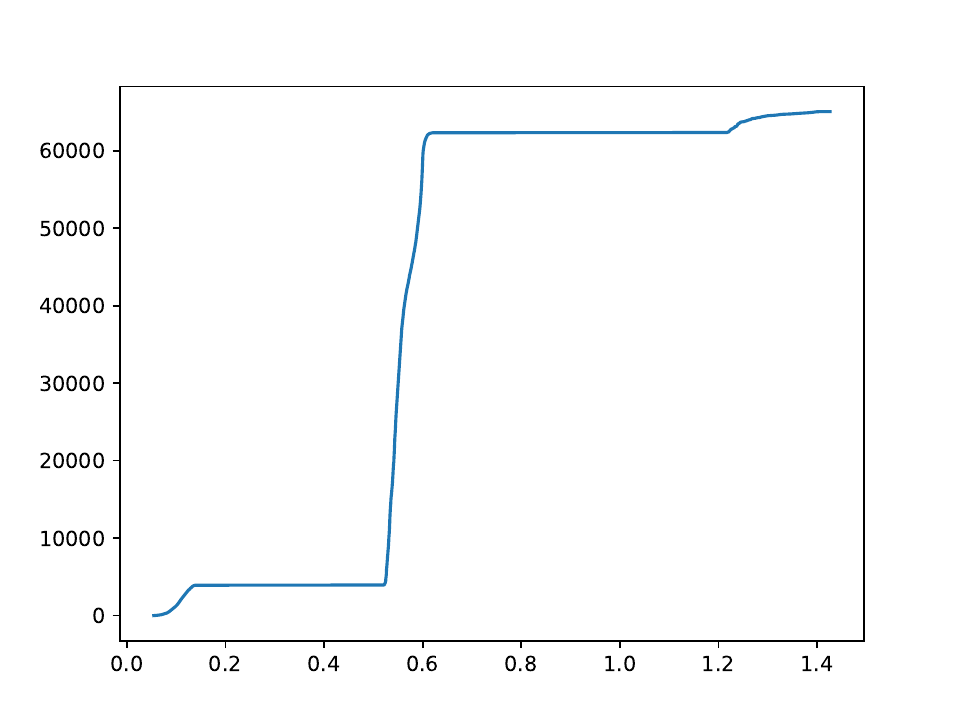}
    \includegraphics[width=.5\textwidth]{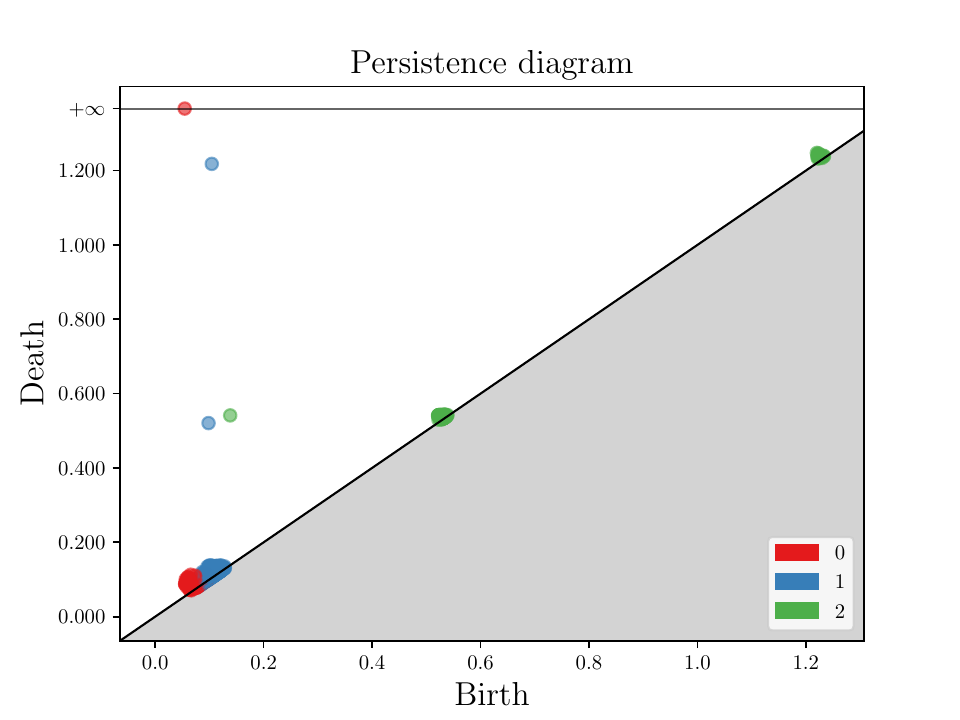}
    \caption{Filtration value of edges for a torus (top). Orange is for original edges and blue after collapse. Top right: enlarged blue graph. Bottom: persistence diagram.}
    \label{fig:torus}
\end{figure}
First, note that some implementations (of which the first one is Eirene~\cite{eirene}) of Rips persistence first check at which filtration value the complex becomes a cone (here around $2$) and ignore longer edges. In our algorithm, this check is performed implicitly  and the long edges are dominated by the apex of the cone and thus get removed (we actually manage to go significantly lower than $2$). Still, it remains sensible to avoid those edges when possible.

After collapsing, we notice several regions in the curve. First some short edges are added progressively, until the complex gets the homotopy type of a torus. Then nothing happens for a while, until we have enough edges to kill one of the 1-cycles and fill the cavity, where many edges are inserted at the same time. Then again nothing happens while the complex is equivalent to a circle, until we can kill this last 1-cycle, and the process quickly stops with a contractible complex.

\subparagraph{Benchmark backward vs forward.} We benchmark the new backward algorithm with the forward algorithm. For the forward algorithm, we use the code from Giotto-ph~\cite{giotto-ph}, which is derived from our implementation in Gudhi but faster by a factor $1.5$ to $2$. Our bench marking considers two aspects: run-time and reduction size (see \cref{fig:speed}). The datasets are: \emph{uniform} for an i.i.d. sample of points in a square, \emph{sparse} for the same, but using a low threshold on the maximal size of edges, \emph{polygon} for a regular polygon, \emph{circle} for an i.i.d. uniform sample of a circle, \emph{dragon} comes from~\cite{NinaPaper} and \emph{O3} from~\cite{ripser} (the first version uses a threshold of $1.4$ on edge lengths).

The backward algorithm comes with an optimization using a \emph{dense} array indexed by vertices. This usually speeds things up nicely, but in cases where the original set of edges is very sparse, this dense array can be an issue, so we also have a version without this array, denoted \emph{sparse}.
\begin{table}[H]
\centering
\begin{tabular}{l|l|l|l|l|l|l|l|}
\cline{4-8}
\multicolumn{3}{c|}{} & \multicolumn{2}{c|}{Forward} & \multicolumn{3}{c|}{Backward} \\
\cline{2-8}
& vertices & before & after & time & after & time dense & time sparse \\
\hline
uniform & 1000 & 499500 & 2897 & 2.4 & 2897 & {\bf 1.7} & 2.4 \\
\hline
sparse & 50000 & 389488 & 125119 & 0.3 & 125119 & 1.9 & {\bf 0.17} \\
\hline
polygon & 300 & 44850 & 44701 & 3.6 & 44701 & {\bf 0.5} & 1 \\
\hline
circle & 300 & 44850 & 41959 & 4.8 & 41959 & {\bf 0.4} & 0.8 \\
\hline
complete & 900 & 404550 & 24540 & 43 & {\bf 5980} & {\bf 0.4} & 0.4 \\
\hline
torus & 1307 & 853471 & 94993 & 31 & 94993 & {\bf 3.2} & 5 \\
\hline
dragon & 2000 & 1999000 & 53522 & 29 & 53522 & {\bf 14} & 20 \\
\hline
O3 (1.4) & 4096 & 4107941 & 13674 & 59 & 13674 & {\bf 37} & 51 \\
\hline
O3 & 1024 & 523776 & 519217 & 200 & 519217 & {\bf 12} & 23 \\
\hline
\end{tabular}
\caption{Run-time and reduction size comparison. Column \textit{before} and \textit{after} contains the number of edges before and after collapsing, and column \textit{time} contains run time in seconds of the collapse.}
\label{fig:speed}
\end{table}

Table~\ref{fig:speed} shows a clear advantage for the backward algorithm in cases where few edges can be removed, or when several edges have the same filtration value. Except for \emph{complete} which is a plain complete graph with every edge at the same filtration value, all edges are computed as Euclidean Rips graphs.

When all the input edges have distinct filtration values, both algorithms output exactly the same list of edges. However, this isn't the case anymore when multiple edges have the same filtration value (and in particular if we apply the algorithm several times). The forward algorithm, as presented, relies on the order of the edges and does not take advantage of edges with the same filtration value. The backward algorithm, at its core, checks if an edge is dominated \emph{at a specific filtration value (grade)}. As seen in Table~\ref{fig:speed}, for a complete graph on 900 vertices, the backward algorithm outputs 5 times fewer edges than the forward algorithm.

\subparagraph{Size gains with approximate version.}
\begin{table}[h]
\centering
\begin{tabular}{c|c|c|c|c|c|c|c|c|}
& original & 1 (exact) & 1.01 & 1.1 & 1.5 & 2 & 10 & 100 \\ \hline
uniform & 499500 & 2897 & 2891 & 2859 & 2609 & 2462 & 2356 & 2353 \\ \hline
circle & 44850 & 42007 & 30423 & 20617 & 17552 & 16404 & 14574 & 14342 \\
\it (seconds) & & \it 0.4 & \it 0.33 & \it 0.22 & \it 0.16 & \it 0.14 & \it 0.12 & \it 0.115 \\ \hline
dragon & 1999000 & 53522 & 52738 & 52161 & 45439 & 40564 & 36094 & 35860 \\ \hline
O3 (1.4) & 4107941 & 13674 & 13635 & 13418 & 12682 & 12050 & 11828 & 11823 \\ \hline
\end{tabular}
\caption{Gains with the approximate algorithm, for different interleaving factors.}
\label{tab:approx}
\end{table}
Table~\ref{tab:approx} shows the number of remaining edges when we don't require the output to have the same persistence diagram, but only ask that the modules be multiplicatively $\alpha$-interleaved. Usually, the approximate version gives modest gains over the exact version, for roughly the same running time. However, in some cases that are hard to simplify like the circle, even a small error allows a significant number of collapses.

\subparagraph{Parallelism benchmark.}
\begin{figure}[h]
    \centering
    \includegraphics[height=8cm]{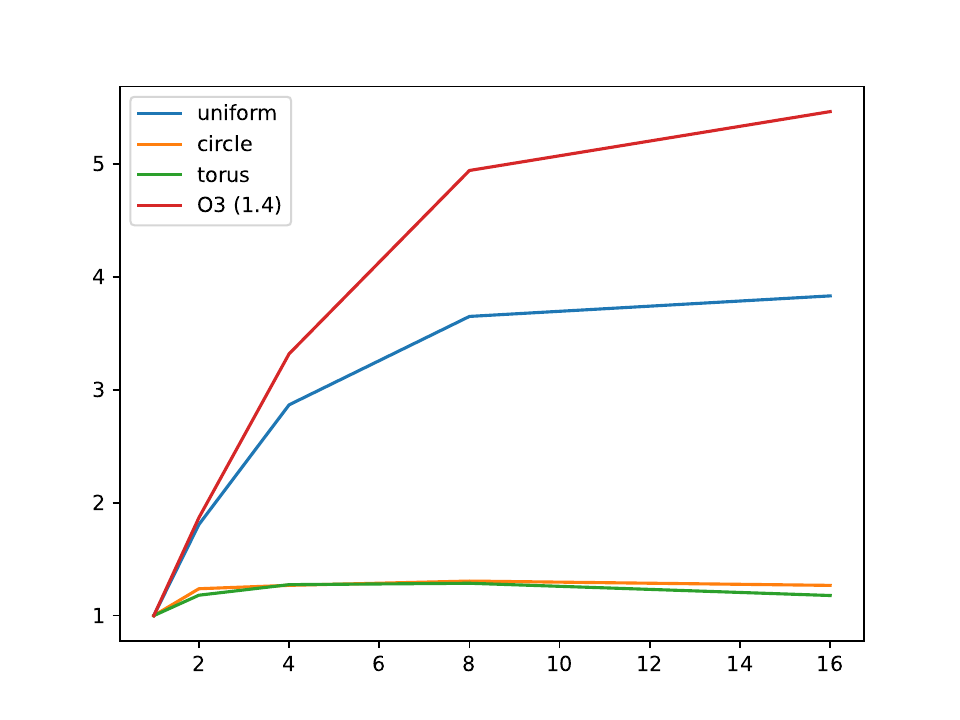}
    \caption{Speed gain in function of the number of threads.}
    \label{fig:parallel}
\end{figure}
We wrote a limited\footnote{This implementation assumes that no two edges have the same filtration value.} prototype based on \lstinline{tbb::parallel_reduce} and tested it on an i7-10875H CPU (8 cores, 16 threads) by limiting the number of threads. \autoref{fig:parallel} shows promising results for some datasets, but also that there is room for better parallel algorithms.

\subparagraph{Persistence benchmark.}

In our experience, doing edge collapses before computing persistent homology helps a lot for (homology) dimension 2 or higher. However, it is a terrible idea if we only care about dimension 0, since computing 0-persistence is cheaper than this simplification can ever hope to be. The case of dimension 1 is more mixed, it can help in some cases and hurt in others. By default we would only recommend its use for dimension  greater than or equal to 2.

For convenience, the persistence computation is done using the version of Ripser~\cite{ripser} found in giotto-ph~\cite{giotto-ph} with $n\_threads=1$, and with our new backward algorithm. This means that edges after the complex has become a cone are ignored. \autoref{tab:ripser} shows the time it takes to compute persistent homology in dimension up to $k$, either directly, or first collapsing before computing it.

\begin{table}[h]
\centering
\begin{tabular}{c|c|c|c|c|c|}
 & dim 1 & collapse \& dim 1 & dim 2 & collapse \& dim 2 & collapse \& dim 3 \\ \hline
torus3D & 6.2 & 3.8 & 75 & 6.4 & 47 \\ \hline
dragon & 3.3 & 9.2 & 148 & 9.7 & 16.3 \\ \hline
\end{tabular}
\caption{Persistent homology computation time in seconds, with or without edge collapse.
}
\label{tab:ripser}
\end{table}

\bibliography{mybib}
\appendix

\label{sec:appendix_comb_shift}

\section{Combinatorial view on shifting} \label{sec:combi_shift}
In this Section, we present a more combinatorial view of the building blocks from~\autoref{sec:shift_swap_trim}, that works directly on a sequence of edges.
Before we proceed, we will fix some notations. Let $\{1, 2, \cdots, n\}$ be a finite index set and $\{t_1, t_2, \cdots, t_n\}\in\R^n$ associated \textit{filtration values} such that for $i < j$, $t_i \leq t_j$. For convenience, we may consider that $t_{n+1}=\infty$. With each $i$ we associate a graph $G_i$
such that $G_{i} \hookrightarrow G_{i+1}$ is an \textit{elementary inclusion}, namely the inclusion of a single edge $e_{i+1}$. The flag complex of $G_{i}$ is denoted as $\overline{G}_{i}$ and we consider the associated flag filtration $\mathcal{F} : \overline{G}_{1} \hookrightarrow \overline{G}_{2} \hookrightarrow \cdots \hookrightarrow \overline{G}_{n}$. The edges in the set $E := \{e_{1}, e_{2}, \cdots e_{n} \}$ are thus indexed with an order compatible with the filtration values. The persistence diagram is defined as the multiset of points $(t_i,t_j)$ corresponding to the half-open intervals $[i,j)$ in the decomposition of the homology module, with the points $(t,t)$ on the diagonal removed\footnote{It is also possible to take all the points on the diagonal with infinite multiplicity, with the same effect.}.

\begin{lemma} [Shifting Lemma] \label{lemma:c_shift}
  If $e_{i}$ is a dominated edge in the graph $G_{i}$. Then, the filtration value $t_i$ can be replaced with any filtration value $t'_i$ such that $t_{i-1}\leq t'_i\leq t_{i+1}$ without changing the persistence diagram.
\end{lemma}
The interval decomposition of the module remains the same, since the module does not change. Because $H(\overline{G}_{i-1})\rightarrow H(\overline{G}_{i})$ is an isomorphism, no non-trivial interval starts or ends\footnote{It is important here that we defined the intervals of the decomposition as half-open.} at index $i$, and thus no point in the diagram is affected by the value of $t'_i$, we only need to preserve the non-decreasing property of $t$. We will usually set $t'_i=t_{i+1}$.

\begin{lemma} [Swapping Lemma] \label{lemma:c_swap}
  If $t_i=t_{i+1}$, we can exchange edges $e_i$ and $e_{i+1}$ without changing the persistence diagram.
\end{lemma}
This defines a new $G'_i=G_{i-1}\cup e_{i+1}$, and although the decomposition may change a bit (some interval extremities may change from $i$ to $i+1$ or vice versa, and intervals $[i,i+1)$ may appear or disappear), the persistence diagram remains the same. This can be seen for instance as a consequence of the detailed analysis of simplex swaps for vineyards~\cite{vineyards}.

\begin{lemma} [Trimming Lemma]\label{lemma:c_trimming}
  If $e_{n}$ is a dominated edge in the graph $G_n$, then reducing the index set to ${1,\ldots,n-1}$ preserves the persistence diagram.
\end{lemma}
This can be seen as a special case of the shifting operation (\autoref{lemma:c_shift}), delaying an edge insertion to infinity.

\section{More on parallelism}\label{sec:appendix_parallel}
The algorithm presented in \cref{sec:parallel} is one simple way to parallelize the algorithm in order to show that it has potential. However, there could be several different possibilities.

For some datasets, \cref{alg:core_flag_filtration} spends most of its time computing the common neighbors of 2 vertices of an edge. The parallelization presented here causes us to recompute these edge neighborhoods several times. It would be possible to store them, but that would increase the memory requirements of the algorithm significantly. One tempting possibility is to use some threads to scout ahead, precomputing a small number of those neighborhoods so they are ready when the main sequential pass needs them (up to some minor updates in case other edges shifted).

For other datasets, most of the time is spent checking if an edge is dominated by a vertex. For instance, in the divide-and-conquer algorithm, if we keep many edges of the right half and remove many edges of the left half, the last merging phase is very costly and sequential. It would be possible, instead of shifting one edge all the way through the second half before considering the next edge, to shift it only half of the way, then passing it to some other thread responsible for the last quarter, forming a pipeline.

It is possible to check for all edges in parallel if they are dominated at their insertion index (giving all of them a different index). The case where no edge can be removed can thus be handled in an embarrassingly parallel fashion. After this parallel detection phase, we can proceed with a quick and simple sequential (or partially parallel) pass that removes the final dominated edges and shifts many edges to the next grade with a non-dominated edge. This can significantly reduce the number of different filtration values, which suggests it could be used as preprocessing before running the plain backward algorithm.

\end{document}